\documentclass[%draft,
a4paper,reqno]{amsart}
\usepackage[T1]{fontenc}
\usepackage[english]{babel}
\usepackage{amssymb,bbm,enumerate}
\usepackage{color}
\usepackage[linktocpage=true,colorlinks=true, linkcolor=blue, citecolor=red, urlcolor=green]{hyperref}
\usepackage[all]{xy}

\renewcommand{\ker}{\Ker}

\newcommand{\mc}[1]{\mathcal{#1}}
\newcommand{\mf}[1]{\mathfrak{#1}}
\newcommand{\mb}[1]{\mathbb{#1}}

\newcommand{\tint}{{\textstyle\int}}

\DeclareMathOperator{\ad}{ad}
\DeclareMathOperator{\Ad}{Ad}
\DeclareMathOperator{\im}{Im}

\DeclareMathOperator{\Ker}{Ker}

\theoremstyle{plain}
\newtheorem{theorem}{Theorem}[section]

\newtheorem{proposition}[theorem]{Proposition}

\theoremstyle{definition}
\newtheorem{definition}[theorem]{Definition}
\newtheorem{example}[theorem]{Example}

\theoremstyle{remark}

\numberwithin{equation}{section}

\definecolor{light}{gray}{.9}

\title{On classical finite and affine $\mc W$-algebras}

\author{Alberto De Sole}

\begin{document}

\pagestyle{plain}

\maketitle

\begin{abstract}
This paper is meant to be a short review and summary of recent results
on the structure of finite and affine classical $\mc W$-algebras,
and the application of the latter to the theory of generalized Drinfeld-Sokolov hierarchies.
\end{abstract}

%%%%%%%%%%%%%%%%%%%%%%%%%%%%%%%%%
\section{Introduction}

In Classical (Hamiltonian) Mechanics the phase space,
describing the possible configurations of a physical system, 
is a Poisson manifold $M$.
The physical observables are the smooth functions on $M$ with real values,
and they thus form a \emph{Poisson algebra} (PA).
The Hamiltonian equations, describing the time evolution of the system,
are written in terms of the Poisson bracket:
$\frac{du}{dt}=\{h,u\}$,
where $h(x)\in C^\infty(M)$ is the Hamiltonian function
(corresponding to the energy observable).

When we quantize a classical mechanic theory we go to Quantum Mechanics.
The observables become non commutative objects,
and the Poisson bracket is replaced by the commutator of these objects.
Hence, the physical observables
in quantum mechanics form an \emph{associative algebra} (AA) $A$.
The phase space is then described as a representation $V$ of $A$,
and the  Schroedinger's equation, describing the evolution of the physical system,
is written in terms of this representation:
$\frac{d\psi}{dt}=H(\psi)$,
where $H\in A$ is the Hamiltonian operator.

Going from a finite to an infinite number of degrees of freedom,
we pass from classical and quantum mechanics to
classical and quantum field theory respectively.
In some sense,
the algebraic structure of the space of observables 
in a conformal field theory is that of a \emph{vertex algebra} (VA) \cite{Bor86},
and its quasi-classical limit is known as \emph{Poisson vertex algebra} (PVA) \cite{DSK06}.

We can summarize the above observations in the following diagram
of the algebraic structures of the four fundamental physical theories:

%%%%%% DIAGRAM %%%%%%%%%%%%%%%%%%%%%%%%%%
\begin{equation}\label{maxi}
\UseTips
\xymatrix{
% 1
PVA\,\,\,\,\,\,  \ar[d]^{\text{Zhu}}
\ar@/^1pc/@{.>}[r]^{\text{quantization}} &
\ar[l]^{\text{cl.limit}}
\,\,\,\,\,\,VA 
\ar[d]_{\text{Zhu}} \\
% 2
PA\,\,\,\,\,\,  \ar@/^1pc/@{.>}[u]^{\text{affiniz.}}
\ar@/_1pc/@{.>}[r]_{\text{quantization}} &
\ar[l]_{\text{cl.limit}}
\,\,\,\,\,\, AA \ar@/_1pc/@{.>}[u]_{\text{affiniz.}}
}
\end{equation}
%%%%%%%%%%%%%%%%%%%%%%%%%%%%%%%%%%%%%%

The arrows in the above diagram have the following meaning.
If we have a filtered associative algebra,
its associated graded is automatically a Poisson algebra called its \emph{classical limit}.
Similarly, if we have a filtered vertex algebra,
its associated graded is a Poisson vertex algebra.
Furthermore, starting from a positive energy vertex algebra 
(respectively Poisson vertex algebra)
we can construct an associative algebra (resp. Poisson algebra)
governing its representation theory, known as its \emph{Zhu algebra}, \cite{Zhu96}.
On the other hand,
the processes of going from a classical theory  to a quantum theory (``quantization''),
or from finitely many to infinitely many degrees of freedom (``affinization''),
do not correspond to canonical functors,
and they are represented in the diagram with dotted arrows.

\medskip

$\mc W$-\emph{algebras} provide a very rich family of examples,
parametrized by a simple Lie algebra $\mf g$ and a nilpotent element $f\in\mf g$,
which appear in all the 4 fundamental aspects in diagram \eqref{maxi}:
%%%%%% DIAGRAM %%%%%%%%%%%%%%%%%%%%%%%%%%
\begin{equation}\label{maxi2}
\UseTips
\xymatrix{
% 1
\mc W^{cl}_z(\mf g,f) 
\,\,\,\,\,\, 
\ar[d]^{\text{Zhu}}
\ar@/^1pc/@{.>}[r]^{\text{quantization}} &
\ar[l]^{\text{cl.limit}}
\,\,\,\,\,\,
\mc W_k(\mf g,f)
\ar[d]_{\text{Zhu}} \\
% 2
\mc W^{cl,fin}(\mf g,f)
\,\,\,\,\,\,  
\ar@/^1pc/@{.>}[u]^{\text{affiniz.}}
\ar@/_1pc/@{.>}[r]_{\text{quantization}} &
\ar[l]_{\text{cl.limit}}
\,\,\,\,\,\, 
\mc W^{fin}(\mf g,f)
\ar@/_1pc/@{.>}[u]_{\text{affiniz.}}
}
\end{equation}
%%%%%%%%%%%%%%%%%%%%%%%%%%%%%%%%%%%%%%

Each of these classes of algebras
was introduced and studied separately, with different applications in mind,
and only later it became fully clear the relations between them.

%%%
\subsubsection*{Classical finite $\mc W$-algebras}
The classical finite $\mc W$-algebra $\mc W^{cl,fin}(\mf g,f)$ 
is a Poisson algebra,
which can be viewed as the algebra of functions 
on the so-called \emph{Slodowy slice} $\mc S(\mf g,f)$.
It was introduced by Slodowy
while studying the singularities 
associated to the coadjoint nilpotent orbits of $\mf g$, \cite{Slo80}.

%%%
\subsubsection*{Finite $\mc W$-algebras}
The first appearance of the finite $\mc W$-algebras $\mc W^{fin}(\mf g,f)$ 
was in a paper of Kostant, \cite{Kos78}.
He constructed the finite $\mc W$-algebra for principal nilpotent $f\in\mf g$
(in which case it is commutative),
and proved that it is isomorphic to the center 
of the universal enveloping algebra $U(\mf g)$.
The construction was then extended in \cite{Lyn79}
for even nilpotent element $f\in\mf g$.
The general definition of finite $\mc W$-algebras $\mc W^{fin}(\mf g,f)$,
for an arbitrary nilpotent element $f\in\mf g$,
appeared much later, \cite{Pre02}.
Starting with the work of Premet, there has been a revival of
interest in finite $\mc W$ -algebras
in connection to geometry and representation theory of simple
finite-dimensional Lie algebras,
and the theory of primitive ideals (see \cite{Mat90,Pre02,Pre05,BrK06}).

%%%
\subsubsection*{Classical $\mc W$-algebras}
The classical (affine) $\mc W$-algebras $\mc W^{cl}_z(\mf g,f)$ 
(depending on the parameter $z\in\mb F$)
were introduced, for principal nilpotent element $f$,
in the seminal paper of Drinfeld and Sokolov \cite{DS85}.
They were introduced as Poisson algebras of function 
on an infinite dimensional Poisson manifold,
and they were used to study KdV-type integrable bi-Hamiltonian hierarchies of PDE's,
nowdays known as Drinfled Sokolov hierarchies.
Subsequently, in the 90's,
there was an extensive literature extending the Drinfeld-Sokolv construction
of classical $\mc W$-algebras
and the corresponding generalized Drinfeld-Sokolv hierarchies
to other nilpotent elements, \cite{dGHM92,FHM92,BdGHM93,DF95,FGMS95,FGMS96}.
Only very recently, in \cite{DSKV13a},
the classical $\mc W$-algebras $\mc W^{cl}_z(\mf g,f)$
were described as Poisson vertex algebras,
and the theory of generalized Drinfeld-Sokolv hierarchies
was formalized in a more rigorous and complete way.

%%%
\subsubsection*{$\mc W$-algebras}
The first (quantum affine) $\mc W$-algebra which appeared in literature
was the so called Zamolodchikov $\mc W_3$-algebra \cite{Zam85},
which is the $\mc W$-algebra associated to $\mf{sl}_3$ and its principal 
nilpotent element $f$.
It was introduced as a ``non-linear'' infinite dimensional Lie algebra
extending the Virasoro Lie algebra,
describing the symmetries of a conformal filed theory.
After the work of Zamolodchikov, a number of papers on
affine $\mc W$-algebras appeared in physics literature, 
mainly as ``extended conformal algebras'',
i.e. vertex algebra extensions of the Virasoro vertex algebra.  
A review
of the subject up to the early 90's may be found in the
collection of a large number of reprints on $\mc W$-algebras \cite{BS95}.
The most important results of this period are in the work by
Feigin and Frenkel \cite{FF90,FF90b}, where the general construction of
$\mc W$-algebras, via a quantization of the Drinfeld-Sokolov reduction,
was introduced in the case of the principal nilpotent element $f$.  
For example, if $\mf g = s\ell_n$, we get the Virasoro vertex algebra for $n=2$, 
and Zamolodchikov's $\mc W_3$ algebra for $n=3$.  
The construction was finally generalized to arbitrary nilpotent element $f$
in \cite{KRW03,KW04,KW05}. In these paper, $\mc W$-algebras were applied to
representation theory of superconformal algebras.

\medskip

A complete understanding of 
the links among the four different appearances 
of $\mc W$-algebras in diagram \eqref{maxi2}
is quite recent.
In \cite{GG02} Gan and Ginzburg 
described the finite $\mc W$-algebras as a quantization 
of the Poisson algebra of functions on the Slodowy slice.
They thus proved that the classical finite $\mc W$-algebra $\mc W^{cl,fin}(\mf g,f)$
can be obtained as the classical limit of the finite $\mc W$-algebra $\mc W^{fin}(\mf g,f)$.

As mentioned earlier, the construction of the $\mc W$-algebra $\mc W_k(\mf g,f)$,
for principal nilpotent element $f$, due to Feigin and Frenkel \cite{FF90},
was obtained as a ``quantization'' of the Drinfeld-Sokolov construction of
the classical $\mc W$-algebra $\mc W^{cl}_z(\mf g,f)$.
But it is only in \cite{DSKV13a}
that the classical $\mc W$-algebra $\mc W^{cl}_z(\mf g,f)$
is described as a Poisson vertex algebra
which can be obtained as classical limit
of the $\mc W$-algebra $\mc W_k(\mf g,f)$.

Furthermore, in \cite{DSK06}
it is proved there that the ($H$-twisted) Zhu algebra $Zhu_H\mc W_k(\mf g,f)$
is isomorphic to the corresponding finite $\mc W$-algebra $\mc W^{fin}(\mf g,f)$.
Hence, their categories of irreducible representations are equivalent.
(This result was independently proved in \cite{Ara07}
for principal nilpotent $f$.)
A similar result for classical $\mc W$-algebras holds as well.
It is also proved in the Appendix of \cite{DSK06} 
(in collaboration with A. D'Andrea, C. De Concini and R. Heluani)
that the quantum Hamiltonian reduction 
definition of finite $\mc W$-algebras
is equivalent to the definition via the Whittaker models,
which goes back to \cite{Kos78}.

\medskip

In the present paper we describe in more detail the ``classical'' part of diagram \eqref{maxi2}:
in Section \ref{sec:2} we describe the Poisson structure of the Slodowy slice
and we introduce the classical finite $\mc W$-algebra $\mc W^{cl,fin}(\mf g,f)$.
In order to describe the affine analogue of it,
we first need to describe the classical finite $\mc W$-algebra
$\mc W^{cl,fin}(\mf g,f)$ as a Hamiltonian reduction,
which is done in Section \ref{sec:2.3}.
By taking the affine analogue of this construction,
we obtain the classical $\mc W$-algebra $\mc W^{cl}_z(\mf g,f)$.
Finally, in Section \ref{sec:3.4}
we describe, following \cite{DSKV13a},
how classical $\mc W$-algebras are used to study
the generalized Drinfeld-Sokolov bi-Hamiltonian hierarchies.

\medskip

The present paper is based on lectures given by the author 
for the conference \emph{Lie superalgberas}, at INdAM, Roma, Italy, in December 2012,
and for the conference 
\emph{Symmetries in Mathematics and Physics}, at IMPA, Rio de Janeiro, Brazil, in June 2013.

%%%%%%%%%%%%%%%%%%%%%%%%%%%%%%%%%
\section{Classical finite $\mc W$-algebras}\label{sec:2}

%%%
\subsection{Poisson manifolds}\label{sec:2.0}

Recall that, by definition,
a \emph{Poisson manifold} is a manifold $M=M^n$
together with a Poisson bracket $\{\cdot\,,\,\cdot\}$
on the algebra of functions $C^\infty(M)$,
making it a Poisson algebra
By the Leibniz rule, we can write the Poisson bracket as
$$
\{f(x),g(x)\}=\sum_{i,j}K_{ij}(x)\frac{\partial f}{\partial x_i}\frac{\partial g}{\partial x_j}\,.
$$
The bivector 
$\eta=\sum_{i,j}K_{ij}(x)\frac{\partial}{\partial x_i}\wedge\frac{\partial}{\partial x_j}
\in\Gamma(\bigwedge^2 TM)$ is the Poisson structure of the manifold.
To every function $h\in C^\infty(M)$ on a Poisson manifold $M$ we associate 
a \emph{Hamiltonian vector field} 
$X_h
=\sum_{i,j=1}^nK(x)_{ij}\frac{\partial h(x)}{\partial x_i}\frac{\partial}{\partial x_j}
=\{h,\cdot\}$,
and the corresponding \emph{Hamiltonian flow} (or evolution):
\begin{equation}\label{20130617:eq1}
\frac{dx}{dt}=\{h,x\}=K(x)\nabla_x h
\,.
\end{equation}
(This is the \emph{Hamiltonian equation} associated to the Hamiltonian function $h$.)
If we start from a point $x\in M$ and we follow all the possible 
Hamiltonian flows \eqref{20130617:eq1} thorough $x$,
we cover the \emph{symplectic leaf} through $x$.
The Poisson manifold $M$ is then
disjoint union of its symplectic leaves: $M=\sqcup_{\alpha}S_\alpha$
(which are symplectic manifolds).

It is natural to ask when a Poisson structure $\eta$ on a Poisson manifold $M$
induces a Poisson structure on a submanifold $N$.
Some sufficient condition is given by the following
\begin{proposition}[Va94]\label{20130519:prop}
Suppose that, for every point $x\in N$, denoting by $(S,\omega)$ the symplectic leave 
of $M$ through $x$, we have
\begin{enumerate}[(i)]
\item
the restriction of the symplectic form $\omega(x)$ of $T_xS$
to $T_xN\cap T_xS$ is non-degenerate;
\item
$N$ is transverse to $S$, i.e. $T_xN+T_xS=T_xM$.
\end{enumerate}
Then, 
the Poisson structure on $M$ naturally induces a Poisson structure on $N$,
and the symplectic leave of $N$ through $x$ is $N\cap S$.
\end{proposition}

If $\mf g$ be a Lie algebra,
the dual space $\mf g^*$ has a natural structure of a Poisson manifold.
Indeed, the Lie bracket $[\cdot\,,\,\cdot]$ on $\mf g$
extends uniquely to a Poisson bracket on the symmetric algebra $S(\mf g)$:
if $\{x_i\}_{i=1}^n$ is a basis of $\mf g$, we have
\begin{equation}\label{20130519:eq1}
\{P,Q\}=\sum_{i,j=1}^n\frac{\partial P}{\partial x_i}\frac{\partial Q}{\partial x_j}[x_i,x_j]
\,.
\end{equation}
We can think at $S(g)$ as the algebra of polynomial functions on $\mf g^*$.
Hence, $\mf g^*$ has an induced structure of a Poisson manifold.
In coordinates, if we think at $\{x_i\}_{i=1}^n$ as linear functions, or local coordinates, on $\mf g^*$,
and we let $\{\xi_i=\frac{\partial}{\partial x_i}\}_{i=1}^n$ be the dual basis of $\mf g^*$,
then, by \eqref{20130519:eq1}, 
the Poisson structure $\eta\in\Gamma(\bigwedge^2 T\mf g^*)$ evaluated at $\xi\in\mf g^*$ is
\begin{equation}\label{20130519:eq2}
\eta(\xi)
=
\sum_{i,j=1}^n\xi([x_i,x_j])\xi_i\wedge\xi_j
=\sum_{j=1}^n\ad^*(x_j)(\xi)\wedge\xi_j
\,\in\wedge^2(\mf g^*)
\,.
\end{equation}
By \eqref{20130519:eq2},
the Hamiltonian vector field associated to $a\in\mf g$ is $\ad^*a$,
and the corresponding Hamiltonian flow through $\xi\in\mf g^*$
is $\Ad^*(e^{ta})(\xi)$.
Hence, 
the symplectic leaves of $\mf g^*$ are the coadjoint orbits $S=\Ad^*G(\xi)$
where $G$ is the connected Lie group with Lie algebra $\mf g$.
The Poisson structure on the coadjoint orbits $\Ad^*G(\xi)$
is known as Kirillov-Kostant Poisson structure.
Its inverse is a symplectic structure. At the point $\xi\in\mf g^*$
it is the following non-degenerate skewsymmetric form 
$\omega(\xi)$ on $\ad^*\mf g(\xi)$:
\begin{equation}\label{20130618:eq1}
\omega(\xi)(\ad^*(a)(\xi),\ad^*(b)(\xi))=\xi([a,b])
\,.
\end{equation}

%%%
\subsection{The Poisson structure on the Slodowy slice}\label{sec:2.1}

Let $\mf g$ be a reductive finite dimensional Lie algebra,
and let $f\in\mf g$ be a nilpotent element.
By the Jacoboson-Morozov Theorem,
$f$ can be included in an $\mf{sl}_2$-triple $\{e,h=2x,f\}$,
see e.g. \cite{CM93}.
Let $(\cdot\,|\,\cdot)$ be a non-degenerated invariant symmetric bilinear form on $\mf g$,
and let $\Phi:\,\mf g\stackrel{\sim}{\to}\mf g^*$ be the isomorphism associated to this 
bilinear form: $\Phi(a)=(a|\,\cdot)$.
We also let $\chi=\Phi(f)=(f|\,\cdot)\in\mf g^*$.

The \emph{Slodowy slice} \cite{Slo80} associated to this $\mf{sl}_2$-triple element is, by definition,
the following affine space
\begin{equation}\label{20130620:eq4}
\mc S=\Phi(f+\mf g^e)
=\big\{\chi+\Phi(a)\,\big|\,a\in\mf g^e\big\}\subset\mf g^*
\,.
\end{equation}

Let $\xi=\Phi(f+r)$, $r\in\mf g^e$, be a given point of the Slodowy slice.
The tangent space to the coadjoint orbit $\Ad^*G(\xi)$ at $\xi$ is
%\begin{equation}\label{20130519:eq11}
$T_\xi(\Ad^*G(\xi))=\ad^*(\mf g)(\xi)=\Phi([f+r,\mf g])$,
%\,,
%\end{equation}
while the tangent space to the Slodowy slice at $\xi$ is
%\begin{equation}\label{20130519:eq12}
$T_\xi(\mc S)
\simeq\Phi(\mf g^e)$.
%\,.
%\end{equation}
%
%
Recalling \eqref{20130618:eq1}, one can check that the assumptions 
of Proposition \ref{20130519:prop} hold, \cite{GG02}:
\begin{enumerate}[(i)]
\item
The restriction of the symplectic form \eqref{20130618:eq1} to $T_\xi(\Ad^*G(\xi))\cap T_\xi(\mc S)$
is non-degenerate.
In other words, if $a\in\mf g$ is such that
$[f+r,a]\in\mf g^e$ and $a\perp[f+r,\mf g]\cap\mf g^e$,
then $a=0$.
\item
The Slodowy slice $\mc S$ intersect transversally the coadjoint orbit at $\xi$,
i.e. $[f+r,\mf g]+\mf g^e=\mf g$.
\end{enumerate}
It then follows by Proposition \ref{20130519:prop}
that $\mc S\subset\mf g^*$ is a Poisson submanifold,
i.e. it has a Poisson structure induced by the Kirillov-Kostant structure on $\mf g^*$.
\begin{definition}\label{20130620:def1}
The \emph{classical finite} $\mc W$ \emph{algebra} $\mc W^{cl,fin}(\mf g,f)\simeq S(\mf g^f)$
is the algebra of polynomial functions on the Slodowy slice $\mc S$
\end{definition}
Clearly, the dual space to $\Phi(\mf g^e)$ is $\mf g^f$.
Hence, 
by the definition \eqref{20130620:eq4} of $\mc S$,
we can identify $\mc W^{cl,fin}(\mf g,f)$,
as a polynomial algebra, with the symmetric algebra over $\mf g^f$.
In fact, we can write down an explicit formula for the Poisson bracket of
the classical finite $\mc W$-algebra.
We have the direct sum decomposition:
$\mf g=[e,\mf g]\oplus\mf g^f$,
and, for $a\in\mf g$, we denote $a^\sharp$ its projection on $\mf g^f$.
Let $\{q_i\}_{i=1}^k$ be a basis of $\mf g^f$
consisting of $\ad x$-eigenvectors,
and let $\{q^i\}_{i=1}^k$ be the dual basis of $\mf g^e$.
For $i\in\{1,\dots,k\}$,
we let $\delta(i)\in\frac12\mb Z$ be $\ad x$-eigenvalue of $q^i$.
By representation theory of $\mf{sl}_2$,
a basis of $\mf g$ is
\begin{equation}\label{20130520:eq6}
\Big\{
q^i_n:=(\ad f)^nq^i\,\Big|\,n=0,\dots,2\delta(i),\,i=1,\dots, k
\Big\}
\,,
\end{equation}
and let $\Big\{q_i^n\,\Big|\,n=0,\dots,2\delta(i),\,i=1,\dots, k\Big\}$,
be the dual basis of $\mf g$.
(Here and further we let $q^i_0=q^i$ and $q_i^0=q^i$.)
\begin{theorem}[\cite{DSK13c}]\label{20130521:prop}
The Poisson bracket on the classical finite $\mc W$-algebra 
$\mc W^{cl,fin}(\mf g,f)$ is ($p,q\in\mf g^f$):
$$
\{p,q\}_{\mc S}=
[p,q]+
\sum_{s=1}^\infty
\sum_{i_1,\dots,i_s=1}^k
\sum_{m_1,\dots,m_s=0}^d
[p,q^{i_1}_{m_1}]^\sharp
[q^{m_1+1}_{i_1},q^{i_2}_{m_2}]^\sharp
\dots
[q^{m_s+1}_{i_s},q]^\sharp
\,.
$$
\end{theorem}
\begin{example}\label{20130519:ex}
If $q\in\mf g^f_0=\mf g^e_0$, 
then $[q^{m+1}_i,q]^\sharp=0$ for all $i,m$.
Hence, $\{p,q\}_{\mc S}=[p,q]\in\mf g^f$.
\end{example}

%%%
\subsection{Classical Hamiltonian reduction}\label{sec:2.2}

In order to define, in Section \ref{sec:3}, the classical $\mc W$-algebra $\mc W^{cl}_z(\mf g,f)$,
i.e. affine analogue of the classical finite $\mc W$-algebra $\mc W^{cl,fin}(\mf g,f)$,
it is convenient to describe the Poisson structure on the Slodowy slice $\mc S$
via a Hamiltonian reduction of $\mf g^*$.
In this section we describe, in a purely algebraic setting,
the general construction of the classical Hamiltonian reduction 
of a Hamiltonian action of a Lie group $N$
on a Poisson manifold $P$.
In the next Section \ref{sec:2.3} we 
then describe the classical finite $\mc W$-algebra $\mc W^{cl,fin}(\mf g,f)$
as a Hamiltonian reduction.

% geometric setting

Recall that 
the classical Hamiltonian reduction is associated to a Poisson manifold $M$,
a Lie group $N$ with a Hamiltonian action on $M$,
and a submanifold $\mc O\subset\mf n^*$
which is invariant by the coadjoint action of $N$.
The corresponding \emph{Hamiltonian reduction} is, by definition,
$\mu^{-1}(\mc O)\big/N$,
where $\mu:\,M\to\mf n^*$
is the moment map associated to the Hamiltonian action of $N$ on $M$.
One shows that,
indeed, $\mu^{-1}(\mc O)\big/N$ 
has a Poisson structure induced by that to $M$, \cite{GS90}.

% algebraic level

On a purely algebraic level, 
going to the algebras of functions,
the classical Hamiltonian reduction can be defined as follows.
%
% SET UP
%
Let $(P,\cdot,\{\cdot\,,\,\cdot\})$ be a unital Poisson algebra.
Let $\mf n$ be a Lie algebra.
Let $\phi:\,\mf n\to P$ be a Lie algebra homomorphism,
and denote by $\phi:\,S(\mf n)\to P$ the corresponding Poisson algebra 
homomorphism.
Let $I\subset S(\mf n)$ be a subset
which is invariant by the adjoint action of $\mf n$,
i.e. such that $\ad(\mf n)(I)\subset I$.
Consider the ideal $P\phi(I)$ of $P$ generated by $\phi(I)$.
Note that, in general, $P\phi(I)$ is NOT a Poisson ideal,
so the quotient space $P/P\phi(I)$ has an induced
structure of a commutative associative algebra,
but NOT of a Poisson algebra.

% DEFINITION

\begin{definition}\label{20130517:def}
The \emph{Hamiltonian reduction}
of the Poisson algebra $P$ associated to the Lie algebra homomorphism $\phi:\,\mf n\to P$
and to the $\ad\mf n$-invariant subset $I\subset S(\mf n)$ is,
as a space,
\begin{equation}\label{20130516:eq2}
\mc W(P,\mf n,I):=
\big(P/P\phi(I)\big)^{\mf n}
=\Big\{ f\in P\,\Big|\,\{\phi(\mf n),f\}\subset P\phi(I)\Big\}\Big/P\phi(I)
\,.
\end{equation}
\end{definition}
\begin{proposition}\label{20130516:prop}
The Hamiltonian reduction $\mc W(P,\mf n,I)$ 
has an induced structure of a Poisson algebra.
\end{proposition}
\begin{proof}
First, it follows by the Leibniz rule that 
$\big\{ f\in P\,\big|\,\{\phi(\mf n),f\}\subset P\phi(I)\big\}\subset P$
is a subalgebra with respect to the commutative associative product of $P$,
and, since by assumption the set $I$ is $\ad(\mf n)$-invariant,
$P\phi(I)\subset\big\{ f\in P\,\big|\,\{\phi(\mf n),f\}\subset P\phi(I)\big\}$ is its ideal.
Hence, their quotient $W(P,\mf n,S)$
has an induced commutative associative product.
We use the same argument for the Poisson structure:
we claim that 
\begin{enumerate}[(i)]
\item
$\big\{ f\in P\,\big|\,\{\phi(\mf n),f\}\subset P\phi(I)\big\}\subset P\phi(I)$
is a Lie subalgebra,
\item
and that 
$P\phi(I)\subset\big\{ f\in P\,\big|\,\{\phi(\mf n),f\}\subset P\phi(I)\big\}$ is its Lie algebra ideal.
\end{enumerate}
Suppose that $f,g\in P$ are such that $(\ad\phi(\mf n))(f)\subset P\phi(I)$ 
and $(\ad\phi(\mf n))(g)\subset P\phi(I)$.
Then, by the Jacobi identity,
$$
\begin{array}{l}
\displaystyle{
\vphantom{\Big(}
\{\phi(\mf n),\{f,g\}\}\subset\{\{\phi(\mf n),f\},g\}+\{f,\{\phi(\mf n),g\}\}
\subset\{P\phi(I),g\}+\{f,P\phi(I)\}
} \\
\displaystyle{
\vphantom{\Big(}
\subset P\{\phi(I),g\}+P\{f,\phi(I)\}+P\phi(I)
\subset P\{\phi(S(\mf n)),g\}+P\{\phi(S(\mf n)),f\}+P\phi(I)
} \\
\displaystyle{
\subset P\{\phi(n),g\}+P\{\phi(\mf n),f\}+P\phi(I)
\subset P\phi(I)
\,.}
\end{array}
$$
In the second inclusion we used the assumption on $f$ and $g$,
in the third inclusion we used the Leibniz rules,
in the fourth inclusion we used the fact that, by construction, $I\subset S(\mf n)$,
in the fifth inclusion we used the Leibniz rules,
and in the last inclusion we used again the assumption on $f$ and $g$.
This proves claim (i).

For claim (ii), let $f\in P$ be such that $\{\phi(\mf n),f\}\subset P\phi(I)$.
We have, with the same line of arguments as above,
$$
\begin{array}{l}
\displaystyle{
\vphantom{\Big(}
\{P\phi(I),f\}\subset P\{\phi(I),f\}+P\phi(I)
\subset P\{\phi(S(\mf n)),f\}+P\phi(I)
} \\
\displaystyle{
\vphantom{\Big(}
\subset P\{\phi(\mf n),f\}+P\phi(I)
\subset P\phi(I)
\,.}
\end{array}
$$
\end{proof}

%%%
\subsection{The Slodowy slice via Hamiltonian reduction}\label{sec:2.3}

We want to describe the classical finite $\mc W$-algebra $\mc W^{cl,fin}(\mf g,f)$
introduced in Section \ref{sec:2.1}
as a Hamiltonian reduction of the Poisson algebra $S(\mf g)$.

% algebraic picture

We have the $\ad x$-eignespace decomposition
$\mf g=\bigoplus_{i\in\frac12\mb Z}\mf g_i$.
Let $\omega$ be the following non-degenerate skewsymmetric bilinear form
on $\mf g_{\frac12}$: 
\begin{equation}\label{20130620:eq1}
\omega(u,v)=(f|[u,v])
\,.
\end{equation}
Let $\ell\subset\mf g_{\frac12}$ be a maximal isotropic subspace.
Consider the nilpotent subalgebra
\begin{equation}\label{20130620:eq2}
\mf n=\ell\oplus\mf g_{\geq1}\subset\mf g
\,.
\end{equation}
Since $\ell\subset\mf g_{\frac12}$ is isotropic w.r.t. the bilinear form \eqref{20130620:eq1},
we have $(f|[\mf n,\mf n])=0$.
Hence, the subset
\begin{equation}\label{20130620:eq3}
I=\big\{n-(f|n)\,\big|\,n\in\mf n\big\}\,\subset S(\mf n)
\,,
\end{equation}
is invariant by the adjoint action of $\mf n$.
Hence, we can consider the corresponding Hamiltonian reduction \eqref{20130516:eq2}
applied to the data $(S(\mf g),\mf n,I)$.
\begin{theorem}[\cite{DSK13c}]\label{20130620:thm}
The classical finite $\mc W$-algebra $\mc W^{cl,fin}(\mf g,f)$
is isomorphic to the Hamiltonian reduction
of the Poisson algebra $S(\mf g)$,
associated to the Lie algebra $\mf n\subset S(\mf g)$ given by \eqref{20130620:eq2},
and the $\ad\mf n$-invariant subset $I\subset S(\mf n)$ in \eqref{20130620:eq3}:
$$
\begin{array}{l}
\displaystyle{
\vphantom{\Big(}
\mc W^{cl,fin}(\mf g,f)
\simeq 
\mc W(S(\mf g),\mf n,I)
} \\
\displaystyle{
=\Big\{ p\in S(\mf g)\,\Big|\,\{\mf n,p\}\subset \langle n-(f|n)\rangle_{n\in\mf n}\Big\}
\Big/
S(\mf g)\langle n-(f|n)\rangle_{n\in\mf n}
}
\end{array}
$$
\end{theorem}

% geometric picture

In geometric terms, the restriction map 
$\mu:\,\mf g^*\to\mf n^*$
is a map of Poisson manifolds (the moment map),
and the corresponding dual map
$\mu^*:\,\mf n\to\mf g$,
is the inclusion map.
The element $\chi=(f|\,\cdot)|_{\mf n}\in\mf n^*$
is a character of $\mf n$, in the sense that $\chi([\mf n,\mf n])=0$
(by the assumption that $\ell$ is maximal isotropic).
Hence, $\chi$ is fixed by the coadjoint action of $N$,
the Lie group of $\mf n$.
We can then consider the Hamiltonian reduction 
of the Poisson manifold $\mf g^*$,
by the Hamiltonian action of $N$ on $\mf g^*$
given by the moment map $\mu$,
associated to the $N$-fixed point $\chi\in\mf g^*$:
$$
\text{Ham.Red.}(\mf g^*,N,\chi)
=\mu^{-1}(\chi)/N
=\Phi(f+\mf n^\perp)/N
\,.
$$
Theorem \ref{20130620:thm}
can be then viewed as the algebraic analogue of the following result of Gan and Ginzburg:
\begin{theorem}[\cite{GG02}]\label{20130620:thm2}
The coadjoint action $N\times\mc S\to\Phi(f+\mf n^\perp)$
is an isomorphism of affine varieties.
The corresponding bijection
$$
\mc S\simeq\mu^{-1}(\chi)/N
=\emph{Ham.Red.}(\mf g^*,N,\chi)
\,,
$$
is an isomorphism of Poisson manifolds.
\end{theorem}

%%%%%%%%%%%%%%%%%%%%%%%%%%%%%%%%%
\section{Classical $W$-algebras}\label{sec:3}

%%%
\subsection{Poisson vertex algebras}\label{sec:3.1}

In this section we introduce the notions
of Lie conformal algebra and of Poisson vertex algebra.
They are, in some sense, 
the ``affine analogue'' of a Lie algebra and of a Poisson algebra respectively.
\begin{definition}\label{20130620:def2}
A \emph{Lie conformal algebra} is a $\mb F[\partial]$-module $R$
with a bilinear $\lambda$-bracket 
$[\cdot\,_\lambda\,\cdot]:\,R\times R\to\mb F[\lambda]\otimes R$
satisfying the following axioms:
\begin{enumerate}[(i)]
\item
sesquilinearity:
$[\partial a_\lambda b]=-\lambda[a_\lambda b]$,
$[a_\lambda \partial b]=(\partial+\lambda)[a_\lambda b]$;
\item
skewsymmetry:
$[a_\lambda b]=-[b_{-\lambda-\partial}a]$
(where $\partial$ is moved to the left);
\item
Jacobi identity:
$[a_\lambda[b_\mu c]]-[b_\mu[a_\lambda c]]=[[a_\lambda b]_{\lambda+\mu}c]$.
\end{enumerate}
\end{definition}
\begin{example}\label{20130620:ex1}
Let $\mf g$ be a Lie algebra with a symmetric invariant bilinear form $(\cdot\,|\,\cdot)$.
The corresponding \emph{current Lie conformal algebra} is,
by definition,
$R=\mb F[\partial]\mf g\oplus\mb F$,
with $\lambda$-bracket ($s$ is a fixed element of $\mf g$):
\begin{equation}\label{cur}
[a_\lambda b]=[a,b]+(a|b)\lambda+z(s|[a,b])\,,
\end{equation}
for $a,b\in\mf g$,
extended to $R$ by saying that $\mb F$ is central, 
and by sesquilinearity.
(The term $(a|b)\lambda$ is a 2-cocycle, defining a central extension,
and $z(s|[a,b])$ is a trivial 2-cocycle.)
\end{example}
\begin{definition}\label{20130620:def3}
A \emph{Poisson vertex algebra} is a commutative associative differential algebra $\mc V$
(with derivation $\partial$)
endowed with a Lie conformal algebra $\lambda$-bracket $\{\cdot\,_\lambda\,\cdot\}$,
satisfying the Leibniz rule
\begin{equation}\label{20130620:eq5}
\{a_\lambda bc\}=\{a_\lambda b\}c+\{a_\lambda c\}b\,.
\end{equation}
Note that by the skewsymmetry axiom and the left Leibniz rule \eqref{20130620:eq5}
we get the right Leibniz rule:
\begin{equation}\label{20130620:eq6}
\{ab_\lambda c\}=\{a_{\lambda+\partial} c\}_\to b+\{b_{\lambda+\partial} c\}_\to a\,,
\end{equation}
where the arrow means that $\partial$ should be moved to the right.
\end{definition}
\begin{example}\label{20130620:ex2}
If $R$ is a Lie conformal algebra,
then $S(R)$ has a natural structure of a Poisson vertex algebra,
with $\lambda$-bracket extending the one in $R$ by the left and right Leibniz rules.
\end{example}
\begin{example}\label{20130620:ex3}
As a special case of Example \ref{20130620:ex2},
consider the current Lie conformal algebra $R=\mb F[\partial]\mf g\oplus\mb F$
associated to the Lie algebra $\mf g$ 
and the symmetric invariant bilinear form $(\cdot\,,\,\cdot)$,
defined in Example \ref{20130620:ex1}.
The \emph{affine Poisson vertex algebra} is,
by definition,
$\mc V_z(\mf g)=S(\mb F[\partial]\mf g)$,
with $\lambda$-bracket \eqref{cur}, extended by sesquilinearity and Leibniz rules.
If $\{u_i\}_{i=1}^n$ is a basis of $\mf g$,
the general formula for the $\lambda$-bracket is:
$$
\begin{array}{l}
\displaystyle{
\vphantom{\Big(}
\{P_\lambda Q\}_z
=
\sum_{i,j=1}^n\sum_{m,n\in\mb Z_+}
\frac{\partial Q}{\partial u_j^{(n)}}(\lambda+\partial)^n
} \\
\displaystyle{
\vphantom{\Big(}
\Big([u_i,u_j]+(u_i|u_j)(\lambda+\partial)+z(s|[u_i,u_j])\Big)
(-\lambda-\partial)^m
\frac{\partial Q}{\partial u_j^{(n)}}(\lambda+\partial)^n
\,.}
\end{array}
$$
This is the ``affine analogue'' of the usual Poisson algebra structure on $S(\mf g)$.
It is a 1-parameter family of Poissson vertex algebras, depending on the parameter $z\in\mb F$.
(Having a 1-parameter family is important for applications to the theory of integrable systems.)
\end{example}

%%%
\subsection{Hamiltonian reduction of Poisson vertex algebras}\label{sec:3.2a}

The classical Hamiltonian reduction construction described in Section \ref{sec:2.2}
has an ``affine analogue'' for Poisson vertex algebras.

Let $(\mc V,\partial,\cdot,\{\cdot\,_\lambda\,\cdot\})$ be a Poisson vertex algebra.
Let $R$ be a Lie conformal algebra.
Let $\phi:\,R\to\mc V$ be a Lie conformal algebra homomorphism,
which we extend to a Poisson vertex algebra 
homomorphism $\phi:\,S(R)\to\mc V$.
Let $I\subset S(R)$ be a subset
which is invariant by the adjoint action of $R$,
i.e. such that $[R_\lambda I]\subset \mb F[\lambda]\otimes I$.
Consider the differential algebra ideal $\langle\phi(I)\rangle_\mc V\subset\mc V$ of $\mc V$
generated by $\phi(I)$.
The quotient space $\mc V/\langle\phi(I)\rangle_\mc V$ has an induced
structure of a (commutative associative) differential  algebra
(but NOT of a Poisson vertex algebra).

% DEFINITION

\begin{definition}\label{20130517:defb}
The \emph{Hamiltonian reduction}
of the Poisson vertex algebra $\mc V$ associated to 
the Lie conformal algebra homomorphism $\phi:\,R\to\mc V$
and to the $R$-invariant subset $I\subset S(R)$ is,
as a space,
\begin{equation}\label{20130516:eq2b}
\begin{array}{l}
\displaystyle{
\mc W(\mc V,R,I):=
\big(\mc V\big/
\langle\phi(I)\rangle_\mc V
\big)^{R}
} \\
\displaystyle{
=\Big\{ f\in\mc V\,\Big|\,\{\phi(a)_\lambda f\}
\in \mb F[\lambda]\otimes\langle\phi(I)\rangle_\mc V \, \forall a\in R
\Big\}\Big/
\langle\phi(I)\rangle_\mc V
\,.}
\end{array}
\end{equation}
\end{definition}
\begin{proposition}\label{20130516:propb}
The Hamiltonian reduction $\mc W(\mc V,R,I)$ 
has an induced structure of a Poisson vertex algebra.
\end{proposition}
\begin{proof}
It is analogue to the proof of Proposition \ref{20130516:prop}.
\end{proof}

%%%
\subsection{Construction of the classical $\mc W$-algebras}\label{sec:3.2}

In analogy with the construction, in Theorem \ref{20130620:thm},
of the classical finite $\mc W$-algebra $\mc W^{cl,fin}(\mf g,f)$
via Hamiltonian reduction,
we define its ``affine analogue'', the classical $\mc W$-algebra $\mc W^{cl}_z(\mf g,f)$,
via a Hamiltonian reduction of the affine Poisson vertex algebra $\mc V_z(\mf g)$
in Example \ref{20130620:ex3}.

Let, as in Section \ref{sec:2.3},
$\mf g$ be a reductive finite dimensional Lie algebra,
with a non-degenerate symmetric bilinear form $(\cdot\,|\,\cdot)$.
Let $(e,h=2x,f)$ be an $\mf{sl}_2$-triple in $\mf g$,
and consider the $\ad x$-eignespace decomposition
$\mf g=\bigoplus_{i\in\frac12\mb Z}\mf g_i$,
Recall the definitions of the bilinear form $\omega$ on $\mf g_{\frac12}$ in \eqref{20130620:eq1},
of the nilpotent subalgebra $\mf n\subset\mf g$ in \eqref{20130620:eq2},
and of the $\ad\mf n$-invariant subset $I=\{n-(f|n)\}_{n\in\mf n}\subset S(\mf n)$,
as in \eqref{20130620:eq3}.
Consider the current Lie conformal algebra $\mb F[\partial]\mf n$
with $\lambda$-bracket $\{a_\lambda b\}=[a,b]$ for $a,b\in\mf n$,
and extended to $R$ by sesquilinearity.
It is obviously a Lie conformal subalgebra of $\mc V_z(\mf g)$.
Hence, the inclusion $S(\mb F[\partial]\mf n)\subset\mc V_z(\mf g)$
is a homomorphism of Poisson vertex algebras.
Furthermore, the set $I=\{n-(f|n)\}_{n\in\mf n}\subset S(\mb F[\partial]\mf n)$
is invariant by the $\lambda$-adjoint action of $\mb F[\partial]\mf n$.
Let $\langle n-(f|n)\rangle_{n\in\mf n}\subset\mc V_z(\mf g)$
be the differential algebra ideal of $\mc V_z(\mf g)$ generated by $I=\{n-(f|n)\}_{n\in\mf n}$.
We can consider the Hamiltonian reduction 
associated to the data $(\mc V_z(\mf g,f),\mb F[\partial]\mf n,I)$.
\begin{definition}[\cite{DSKV13a}]\label{20130621:def}
The \emph{classical} $\mc W$-\emph{algebra} is
the following $1$-parameter family (parametrized by $z\in\mb F$)
of Poisson vertex algebras
$$
\begin{array}{l}
\displaystyle{
\vphantom{\Big(}
\mc W^{cl}_z(\mf g,f)
\simeq 
\mc W(\mc V_z(\mf g),\mb F[\partial]\mf n,I)
} \\
\displaystyle{
=\Big\{ P\in S(\mb F[\partial]\mf g)\,\Big|\,\{a_\lambda P\}_z\in \mb F[\lambda]\otimes\langle n-(f|n)\rangle_{n\in\mf n}
\,\forall a\in\mf n\Big\}
\Big/
\langle n-(f|n)\rangle_{n\in\mf n}
\,.}
\end{array}
$$
\end{definition}

It is convenient to give a description of the $\mc W$-algebra as a subspace,
rather than as a quotient space.
Let $\ell^\prime\subset\mf g_{\frac12}$ 
be a maximal isotropic subspace (w.r.t. $\omega$) 
complementary to $\ell$.
Hence, $\mf p=\ell^\prime\oplus\mf g_{\leq 0}\subset\mf g$ is a subspace complementary 
to $\mf n$:
$$
\mf g=\mf n\oplus\mf p=\mf n^\perp\oplus\mf p^\perp\,.
$$
(We can thus identify $\mf p^*=\mf n^\perp$ and $\mf n^*=\mf p^\perp$.)
Consider the differential algebra homomorphism
$\rho:\,\mc V_z(\mf g)\to S(\mb F[\partial]\mf p)$ given by
$$
\rho(a)=\pi_{\mf p}(a)+(f|a)
\,\,\,\,\forall a\in\mf g
\,.
$$
Clearly, $\ker(\rho)=\langle n-(f|n)\rangle_{n\in\mf n}$.
Then, the classical $\mc W$-algebra can be equivalently be defined as follows:
$$
\mc W^{cl}_z(\mf g,f)
\simeq
\Big\{P\in S(\mb F[\partial]\mf p)
\,\Big|\,
\rho(\{a_\lambda P\}_z)=0\,\,\forall a\in\mf n\Big\} 
\,,
$$
with $\lambda$-bracket
$$
\{P_\lambda Q\}_{z,\rho}=\rho\{P_\lambda Q\}_z
\,.
$$

We can find explicit formulas for the $\lambda$-brackets
of generators of $\mc W^{cl}(\mf g,f)$,
analogue to the result of Theorem \ref{20130521:prop}.
\begin{theorem}[\cite{DSK13c}]\label{20130521:propb}
As a differential algebra, the classical $\mc W$-algebra $\mc W^{cl}_z(\mf g,f)$
is isomorphic to the algebra of differential polynomials over $\mf g^f$,
i.e. $\mc W^{cl}_z(\mf g,f)\simeq S(\mb F[\partial]\mf g^f)$.
Moreover,
the $\lambda$ bracket on $\mc W^{fin}_z(\mf g,f)$ is,
using the notation in Theorem \ref{20130521:prop} ($q_{i_0},q_{j_0}\in\mf g^f$):
\begin{equation}\label{monster-eq}
\begin{array}{l}
\displaystyle{
\{p_\lambda q\}_z=
\sum_{s,t=0}^\infty
\sum_{i_1,\dots,i_s=1}^k
\sum_{j_1,\dots,j_t=1}^k
\sum_{m_1,\dots,m_s=0}^d
\sum_{n_1,\dots,n_t=0}^d
} \\
\displaystyle{
\vphantom{\Big(}
\big([q_{j_0},q^{j_1}_{n_1}]^\sharp
-\delta_{n_1,0}\delta_{j_1,j_0}(\lambda+\partial)\big)
\dots
\big([q_{j_{t-1}}^{n_{t-1}+1},q^{j_t}_{n_t}]^\sharp
-\delta_{n_t,n_{t-1}+1}\delta_{j_t,j_{t-1}}(\lambda+\partial)\big)
} \\
\displaystyle{
\vphantom{\Big(}
\big(
[q_{i_s}^{m_s+1},q_{j_t}^{n_t+1}]^\sharp
+(q_{i_s}^{m_s+1}|q_{j_t}^{n_t+1})(\lambda+\partial)
+z(s|[q_{i_s}^{m_s+1},q_{j_t}^{n_t+1}]
\big)
} \\
\displaystyle{
\vphantom{\Big(}
\big([q_{i_{s-1}}^{m_{s-1}+1},q^{i_s}_{m_s}]^\sharp
+\delta_{m_s,m_{s-1}+1}\delta_{i_s,i_{s-1}}(\lambda+\partial)\big)
\dots
\big([q_{i_0},q^{i_1}_{m_1}]^\sharp
+\delta_{m_1,0}\delta_{i_1,i_0}\lambda\big)
\,.}
\end{array}
\end{equation}
(In the RHS, when $s$ or $t$ is $0$,
we replace $m_0+1$ or $n_0+1$ by $0$.)
\end{theorem}
In the special case classical $\mc W$-algebras for principal and minimal
nilpotent element, 
equation \eqref{monster-eq}
was proved in \cite{DSKV13b}.

%%%
\subsection{Application of Poisson vertex algebras to the theory of Hamiltonian equations}
\label{sec:3.3}

Poisson vertex algebras can be used in the study 
of Hamiltonian partial differential equations in classical field theory,
and their integrability \cite{BDSK09}
(in the same way as Poisson algebras are used to study Hamiltonian equations 
in classical mechanics).

The basic observation is that, if $\mc V$ is a Poisson vertex algebra 
with $\lambda$-bracket $\{\cdot\,_\lambda\,\cdot\}$,
then $\mc V/\partial\mc V$ is a Lie algebra, with Lie bracket
$$
\{\tint f,\tint g\}=\tint \{f_\lambda g\}\big|_{\lambda=0}\,,
$$
and we have a representation of the Lie algebra $\mc V/\partial\mc V$ on $\mc V$, 
with the following action
$$
\{\tint f,g\}=\{f_\lambda g\}\big|_{\lambda=0}\,,
$$
We can then introduce Hamiltonian equations
and integrals of motion in the same way as in classical mechanics.

\begin{definition}
Given a Poisson vertex algebra $\mc V$ with $\lambda$-bracket $\{\cdot\,_\lambda\,\cdot\}$,
the \emph{Hamiltonian equation} with Hamiltonian functional $\tint h\in\mc V/\partial\mc V$
is:
\begin{equation}\label{hameq}
\frac{du}{dt}=\{h_\lambda u\}\big|_{\lambda=0}
\,.
\end{equation}
An \emph{integral of motion} for the Hamiltonian equation \eqref{hameq}
is an element $\tint g\in\mc V$ such that
$$
\{\tint h,\tint g\}=\tint \{h_\lambda g\}\big|_{\lambda=0}=0
\,.
$$
The element $g\in\mc V$ is then called a \emph{conserved density}.
The usual requirement to have \emph{integrability}
is to have an infinite sequence $\tint g_0=\tint h,\tint g_1,\tint g_2,\dots$
of linearly independent integrals of motion in involution:
$$
\tint \{{g_m}_\lambda {g_n}\}\big|_{\lambda=0}=0
\,\,\text{ for all } m,n\in\mb Z_+
\,.
$$
\end{definition}
\begin{example}\label{20130621:ex}
The famous \emph{KdV equation}, describing the evolution of waves in shallow water is
$$
\frac{\partial u}{\partial t}
=3u\frac{\partial u}{\partial x}+c\frac{\partial^3 u}{\partial x^3}
\,.
$$
It is a bi-Hamiltonian equation,
since it can be written in two compatible Hamiltonian forms:
$$
\frac{du}{dt}
=
\Big\{\frac12(u^3+cuu'')_\lambda u\Big\}_0\Big|_{\lambda=0}
=
\Big\{\frac12 u^2\,_\lambda u\Big\}_1\Big|_{\lambda=0}
\,,
$$
on the differential algebra $\mc V=S(\mb F[\partial]u)$,
with PVA $\lambda$-brackets
$$
\{u_\lambda u\}_0=\lambda
\,,\,\,
\{u_\lambda u\}_1=u'+2u\lambda+c\lambda^3
\,.
$$
Compatibility means that 
$\{\cdot\,_\lambda\,\cdot\}_z=\{\cdot\,_\lambda\,\cdot\}_0+\{\cdot\,_\lambda\,\cdot\}_1$
is a 1-parameter family of PVA $\lambda$-brackets.
\end{example}

The usual ``trick'' to construct a sequence 
$\tint g_n,\,n\in\mb Z_+$, of integrals of motion in involution
is the so called \emph{Lenard-Magri} scheme:
assuming we have a bi-Hamiltonian equation
$$
\frac{du}{dt}
=
\big\{{h_1}_\lambda u\big\}_0
=
\big\{{h_0}_\lambda u\big\}_1
\,,
$$
we try to solve the following recursion equation for $\tint g_n,\,n\geq0$
(starting with $g_0=h_0$ and $g_1=h_1$),
\begin{equation}\label{lm}
\big\{{g_0}_\lambda u\big\}_0=0
\,\,,\,\,\,\,
\big\{{g_{n+1}}_\lambda u\big\}_0
=
\big\{{g_n}_\lambda u\big\}_1
\,.
\end{equation}
(There are various ``cohomologcal'' arguments indicating that, often,
such recursive equations can be solved for every $n$, see e.g. \cite{DSK12a,DSK12b}.)
In this case, it was a simple observation of Magri \cite{Mag78}
that the solutions $\tint g_n,\,n\in\mb Z_+$,
are integrals of motion in involution
w.r.t. both PVA $\lambda$-brackets $\{\cdot\,_\lambda\,\cdot\}_0$ and $\{\cdot\,_\lambda\,\cdot\}_1$,
and therefore we get the integrable hierarchy of bi-Hamiltonian equations
$$
\frac{du}{dt_n}=\big\{{g_n}_\lambda u\big\}_0\,.
$$

%%%
\subsection{Generalized Drinfeld-Sokolov bi-Hamiltonian integrable hierarchies}
\label{sec:3.4}

Following the ideas of \cite{DS85},
we can prove that the Lenard-Magri scheme can be applied
to construct integrable hierarchies of bi-Hamiltonian equations
attached to the classical $\mc W$-algebras $\mc W^{cl,fin}_z(\mf g,f)$.
For example,
for $\mf g=\mf{sl}_2$, we get the KdV hierarchy (cf. Example \ref{20130621:ex}).

The basic assumption is that there exists
a homogeneous (w.r.t. the $\ad x$-eigenspace decomposition)
element $s\in\ker(\ad\mf n)$ such that $f+s$ is a semisimple element of $\mf g$.
Hence, $f+zs$ is a semisimple element of $\mf g((z^{-1}))$,
and we have the direct sum decomposition
$\mf g((z^{-1}))=\mf h\oplus\mf h^\perp$, where
\begin{equation}\label{20130621:eq1}
\mf h:=\Ker\ad(f+zs)
\,\,\text{ and }\,\,
\mf h^\perp:=\im\ad(f+zs)
\,.
\end{equation}
We define a $\frac12\mb Z$-grading of $\mf g((z^{-1}))$
by letting $\deg(z)=-d-1$, if $s\in\mf g_d$.
In particular, $f+zs$ is homogenous of degree $-1$.
We have the induced decompositions
of $\mf h$ and $\mf h^\perp$ (since $f+zs$ is homogeneous):
\begin{equation}\label{20130621:eq2}
\mf h=\widehat\bigoplus_{i\in\frac12\mb Z}\mf h_i
\,\,\text{ and }\,\,
\mf h^\perp=\widehat\bigoplus_{i\in\frac12\mb Z}\mf h^\perp_i
\,.
\end{equation}
Consider the Lie algebra 
$$
\tilde{\mf g}=
\mb F\partial\ltimes\big(\mf g((z^{-1}))\otimes\mc V(\mf p)\big)
\,,
$$
where $\partial$ acts only on the second factor of the tensor product.
Clearly,
$\mf g((z^{-1}))_{>0}\otimes\mc V(\mf p)\subset\tilde{\mf g}$
is a pro-nilpotent subalgebra.
Hence, for $U(z)\in \mf g((z^{-1}))_{>0}\otimes\mc V(\mf p)$,
we have a well defined automorphism
$e^{\ad U(z)}$ of $\tilde{\mf g}$.
Let $\{q_i\}_{i\in J}$ be a basis of $\mf p$,
and let $\{q^i\}_{i\in J}$ be the dual basis of $\mf n^\perp$.
\begin{theorem}[\cite{DSKV13a}]
\begin{enumerate}[(a)]
\item
There exist unique
$U(z)\in\mf h^\perp_{>0}\otimes\mc V(\mf p)$
and $h(z)\in\mf{h}_{>-1}\otimes\mc V(\mf p)$ such that
\begin{equation}\label{L0_dsr}
e^{\ad U(z)}\Big(\partial+(f+zs)\otimes1+\sum_{i\in J}q^i\otimes q_i\Big)
=\partial+(f+zs)\otimes1+h(z)\,.
\end{equation}
\item
For $0\neq a(z)\in Z(\mf h)$, 
the coefficients $g_n,\,n\in\mb Z_+$,
of the Laurent series 
\begin{equation}\label{gz}
g(z)=(a(z)\otimes1 | h(z))\,\in\mc V(\mf p)((z^{-1}))\,,
\end{equation}
lie in $\mc W^{cl,fin}_z(\mf g,f)\subset\mc V(\mf p)$ modulo $\partial\mc V(\mf p)$,
and they satisfy the Lenard-Magri recursion equations \eqref{lm}
for the PVA $\lambda$-brackets
$\{\cdot\,_\lambda\,\cdot\}_0=\{\cdot\,_\lambda\,\cdot\}_{z=0}$
and
$\{\cdot\,_\lambda\,\cdot\}_1=\frac{d}{dz}\{\cdot\,_\lambda\,\cdot\}_z\big|_{z=0}$
\end{enumerate}
\end{theorem}
Hence, 
we get an integrable hierarchy of bi-Hamiltonian equations,
called the \emph{generalized Drinfeld-Sokolov hierarchy} ($w\in\mc W$),
$$
\frac{dw}{dt_n}=\rho\{{g_n}_\lambda w\}_0\big|_{\lambda=0}
\,\,,\,\,\,\,n\in\mb Z_+\,.
$$

%%%%%%%%%%% BIBLIOGRAFIA %%%%%%%%%%%%%%%%%%%%%%%%

\end{document}